\documentclass[11pt,twoside]{article}

\usepackage{a4}

\usepackage{amssymb,amsmath,amsthm,latexsym}
\usepackage{amsfonts}
  \usepackage{amsfonts}
\usepackage{graphicx}

\usepackage{mathtools,amsthm}
\newtheoremstyle{case}{}{}{}{}{}{:}{ }{}

\usepackage{hyperref}
\usepackage{amsmath, amsfonts}
\usepackage{amssymb, graphicx}
\usepackage{amscd}
\usepackage{textcomp}
\usepackage{palatino}
\usepackage{mathtools}
\usepackage{enumitem}
\usepackage[dvipsnames]{xcolor}

\newtheorem{theorem}{Theorem}[section]

\newtheorem{definition}[theorem]{Definition}
\newtheorem{example}[theorem]{Example}

\newtheorem{proposition}[theorem]{Proposition}
\newtheorem{remark}[theorem]{Remark}

\numberwithin{subcase}{case}

\vspace{5cm}

\begin{document}
  
  \label{'ubf'}  
\setcounter{page}{1}                                 

\markboth {\hspace*{-9mm} \centerline{\footnotesize \sc
         A note on MDS Property of Circulant Matrices
         }
                 }
                { \centerline                           {\footnotesize \sc                     
      T. Chatterjee and A. Laha } \hspace*{-9mm}              
               }

\vspace*{-2cm}

\begin{center}
{ 
       { \textbf {
       A Note on MDS Property of Circulant Matrices
                               }
       }
\\

\medskip
{\sc Tapas Chatterjee }\\
{\footnotesize Indian Institute of Technology Ropar, Punjab, India.
}\\
{\footnotesize e-mail: {\it tapasc@iitrpr.ac.in}}
\medskip

{\sc Ayantika Laha }\\
{\footnotesize Indian Institute of Technology Ropar, Punjab, India.
}\\
{\footnotesize e-mail: {\it 2018maz0008@iitrpr.ac.in}}
\medskip
}
\end{center}

\thispagestyle{empty} 
\vspace{-.4cm}

\hrulefill

\begin{abstract}  
In 2014, Gupta and Ray proved that the  circulant involutory matrices over the finite field $\mathbb{F}_{2^m}$ can not be maximum distance separable (MDS). This non-existence also extends to circulant orthogonal matrices of order $2^d \times 2^d$ over finite fields of characteristic $2$. These findings inspired many authors to generalize the circulant property for constructing lightweight MDS matrices with practical applications in mind.
Recently, in $2022,$ Chatterjee and Laha initiated a study of circulant matrices by considering semi-involutory and semi-orthogonal properties.  Expanding on their work, this article delves into circulant matrices possessing these characteristics over the finite field $\mathbb{F}_{2^m}.$ Notably, we establish a correlation between the trace of associated diagonal matrices and the MDS property of the matrix. We prove that this correlation holds true for even order semi-orthogonal matrices and semi-involutory matrices of all orders. Additionally, we provide examples that for circulant, semi-orthogonal matrices of odd orders over a finite field with characteristic 2, the trace of associated diagonal matrices may possess non-zero values.

\end{abstract}
\hrulefill



\vspace{-.37cm}
\section{Introduction}
Maximum Distance Separable (MDS) matrices have gained significant attention since their implementation in the block cipher SHARK \cite{RDPB} in $1996$. These matrices play a vital role in ensuring optimal diffusion in the diffusion layer of the cipher, thereby enhancing the cipher's resilience against both differential and linear cryptanalysis. As a consequence, many block ciphers such as SQUARE \cite{DKR}, AES \cite{DR}, Twofish \cite{SKWWHF}, and hash functions like PHOTON \cite{GPP}, Whirlpool \cite{BR} incorporate MDS matrices for enhancing the overall security of these cryptographic systems.
 
Two primary approaches for constructing MDS matrices are recursive and non-recursive. In recursive construction, we mainly start with a sparse matrix $A$ of order $n \times n$ with entries from a finite field such that $A^n$ is an MDS matrix. Notably, the block cipher LED \cite{GPPR} and the hash function PHOTON use recursive MDS matrices, which are constructed from companion matrices. On the other hand, in non-recursive constructions, various types of matrices are employed to directly construct MDS matrices. This involves leveraging the algebraic properties of the matrix and making appropriate choices for the entries from the finite fields.
 
 In $1996,$ Youssef {\it{et al.}} \cite{YMT} used a Cauchy matrix to construct an involutory MDS matrix. The involutory property holds significance in block cipher implementation, ensuring that if an MDS matrix $M$ is used for encryption, the same matrix is applied for decryption as well. So, it is of special interest to find efficient MDS matrices which are also involutory or orthogonal. Subsequently, in $2004,$ Lacan and Fimes \cite{LF} first constructed an MDS matrix using two Vandermonde matrices. Later in $2012,$ Sajadieh {\it{et al.}} \cite{SDMO} demonstrated that the Vandermonde based MDS matrix construction proposed by Lacan and Fimes in \cite{LF} could be transformed into an involutory matrix.
Subsequently, numerous researchers \cite{GR1, GPRS} have proposed diverse constructions of MDS matrices, leveraging the unique properties of Cauchy and Vandermonde matrices. 

Another class of matrices that gained a significant attention in the construction of MDS matrices is the circulant matrix. The block cipher AES used a $4 \times 4$ circulant MDS matrix in its diffusion layer, chosen for its lightweight characteristics. The matrix is circulant$(\alpha, \alpha+1, 1, 1)$ with entries from the finite field $\mathbb{F}_{2^8}$, where $\alpha$ is a root of the polynomial $x^8+x^4+x^3+x+1$. The lightweight nature of this matrix is attributed to two ones in the first row, together with the observation that a circulant matrix of order $n \times n$ can have at most $n$ distinct entries. However, it is important to note that this matrix is neither involutory nor orthogonal. In $2014,$ Gupta and Ray \cite{GR2, GR} proved that the circulant orthogonal matrix of order $2^d \times 2^d$ can not be MDS when their entries are from a finite field of characteristic $2$. Furthermore, they demonstrated that circulant matrices of order $n \geq 3$ can not simultaneously possess involutory and MDS properties over the finite field $\mathbb{F}_{2^m}.$  

These two non-existence results inspired authors to study circulant like matrices with MDS property in more detail. In \cite{SS1, SS2}, Sarkar {\it{et al.}} explored Toeplitz matrices, investigating their involutory, orthogonal, and MDS properties. Hankel matrices exhibiting these characteristics were examined in  \cite{GPRS}. In \cite{LS}, Liu and Sim generalized circulant matrices to cyclic matrices by changing the permutations and demonstrated the potential for left-circulant matrices to be MDS and involutory. In $2022,$ Chatterjee and Laha first initiated the study of circulant matrices with semi-involutory and semi-orthogonal properties with entries from a finite field in \cite{TA, TA1}. 
Subsequently, in \cite{TA2, TA3}, Chatterjee {\it{et al.}} studied cyclic matrices and $g$-circulant matrices, meticulously analyzing their involutory, orthogonal, semi-involutory, semi-orthogonal and MDS characteristics.

Many authors have continued the search for MDS matrices from finite fields to rings and modules. In $1995$, Zain and Rajan defined MDS codes over cyclic groups \cite{ZR} and Dong {\it{et al.}} characterized MDS codes over elementary Abelian groups \cite{DCG}. By considering a finite Abelian group as a torsion module over a PID, Chatterjee {\it{et al.}} proved some non-existent results of MDS matrices in $2022$ \cite{TAS}. In \cite{CL}, Cauchois {\it{et al.}} introduced $\theta$-circulant matrices over the quasi-polynomial ring. They proposed a construction for $\theta$-circulant almost-involutory MDS matrix over the quasi-polynomial ring.

\section{Contribution}

In Section \ref{sec:circ semi-ortho} of this article, we explore circulant semi-orthogonal matrices over the finite field $\mathbb{F}_{2^m}$. We establish that, for an order of $2^d \times 2^d,$ the associated diagonal matrices of a circulant semi-orthogonal matrix have trace zero. For other even orders, matrices are classified based on whether the order $k$ is congruent to $0/2 \pmod 4$. Specifically, for $k \equiv 0 \pmod 4$, we demonstrate that if the circulant semi-orthogonal matrix possesses the MDS property, the trace of the associated diagonal matrices is guaranteed to be zero. For other even orders, an additional condition on the diagonal matrices is necessary for the trace to be zero. This section also provides examples of circulant, semi-orthogonal MDS matrices of odd orders.
Similarly, in Section \ref{sec:circ semi-inv}, we present analogous results for circulant semi-involutory matrices of order $n \geq 3$.

\section{Preliminaries} 
 In this section, we describe the notations and important definitions we use throughout the paper. 
 
We begin with some notations and definitions from \cite{MS}.
Let $\mathbb{F}_q$ denote a finite field with $q$ elements where $q$ is power of a prime $p$. Let $\mathcal{C}$ be an $[n, k, d]$ linear error correcting code over the finite field $\mathbb{F}_q$ with length $n$, dimension $k$, and minimum Hamming distance $d$. The code $\mathcal{C}$ is a $k$ dimensional subspace of $\mathbb{F}_q^n$. The generator matrix $G$ of $\mathcal{C}$ is a $k \times n$ matrix with the standard form $[I|A]$, where $I$ is a $k \times k$ identity matrix and $A$ is $k \times n-k$ matrix.  The Singleton bound states that, for a $[n,k,d]$ code, $n-k \geq d-1$. An $[n,k,n-k+1]$ code is called a maximum distance separable (MDS) code. Another definition of an MDS code in terms of a generator matrix is the following.

\begin{definition}
An $[n,k,d]$ code $\mathcal{C}$ with the generator matrix $G=[I|A]$, where $A$ is a $k \times (n-k)$ matrix, is MDS if and only if every $i \times i$ submatrix of $A$ is non-singular, $i=1,2,\hdots,\text{min}(k,n-k)$.
\end{definition}
 This definition of MDS code gives the following characterization of an MDS matrix.

\begin{definition}
A square matrix $A$ is said to be MDS if every square submatrix of $A$ is non-singular.
\end{definition}

MDS matrices with easily implementable inverses  play a crucial role in the decryption layer of an SPN based block cipher. Therefore, either the involutory or orthogonal property of MDS matrices is essential. Here $A^{-1}$ denotes the inverse of $A$, $A^{T}$ denotes the transpose of $A$ and $I$ is the identity matrix.

\begin{definition}\label{semi-ortho def}
A square matrix $A$ is said to be involutory if $A^2=I$ and orthogonal if $AA^T=A^TA=I$. Also the trace of a square matrix is the sum of the elements in the main diagonal.
\end{definition}

In $2012,$ Fielder {\it{et al.}} generalized the orthogonal property of matrices to semi-orthogonal in \cite{FH}. The definition of a semi-orthogonal matrix is as follows.

\begin{definition}
A non-singular matrix $M$ is semi-orthogonal if there exist non-singular diagonal matrices $D_1$ and $D_2$ such that
$M^{-T}=D_1MD_2$, where $M^{-T}$ denotes the transpose of the matrix $M^{-1}$.
\end{definition}
We refer to the matrices $D_1$ and $D_2$ in Definition \ref{semi-ortho def} as associated diagonal matrices for the semi-orthogonal matrix $A.$

Following that, in $2021,$ Cheon {\it{et al.}} \cite{CCK}  defined semi-involutory matrices as a generalization of the involutory matrices. The definition of a semi-involutory matrix is as follows.
\begin{definition}\label{semi-inv def}
A non-singular matrix $M$ is said to be semi-involutory if there exist non-singular diagonal matrices $D_1$ and $D_2$ such that $M^{-1} = D_1MD_2$.  
\end{definition}

We refer to the matrices $D_1$ and $D_2$ in Definition \ref{semi-inv def} as associated diagonal matrices for the semi-involutory matrix $A$.

As previously mentioned, circulant matrices find application in the diffusion layer. In this context, we now provide definitions for circulant matrices and their generalizations. Let $A$ be an $n\times n$ matrix. The $i$-th row of $A$ is denoted by $R_i$ for $0 \leq i \leq n-1$ and the $j$-th column as $C_j$ for $0 \leq j \leq n-1$. Furthermore, $A[i,j]$ denotes the entry at the intersection of the $i$-th row and $j$-th column. The definition of a circulant matrix is the following. 

\begin{definition}
The square matrix of the form $\begin{bmatrix}
c_0 & c_1 & c_2 & \cdots & c_{n-1}\\
c_{k-1} & c_0 & c_1 & \cdots & c_{n-2}\\
\vdots & \vdots & \vdots & \cdots & \vdots\\
c_1 & c_2 & c_3 & \cdots & c_0
\end{bmatrix}$ is said to be circulant matrix and denoted by $\mathcal{C}=$ circulant$(c_0 , c_1 , c_2 , \hdots , c_{n-1})$.
\end{definition}
 The entries of the circulant matrix $\mathcal{C}$can be expressed as $\mathcal{C}[i,j]=c_{j-i},$ where subscripts are calculated modulo $n$.

The determinant of a circulant matrix of order $n \times n$ is 
\begin{eqnarray}\label{det circulant}
 \det(\mathcal{C})=\prod_{j=0}^{n-1}(\sum_{l=0}^{n-1}c_l\omega_n^{jl}),
 \end{eqnarray}
where $\omega_n= e^{\frac{2\pi i}{n}} \in \mathbb{C}$. 

In \cite{TA}, Chatterjee {\it{et al.}} proved the following result for circulant semi-involutory matrices over a finite field.
\begin{theorem}\label{circ semi-inv}
Let $A$ be an $n \times n$ circulant matrix over a finite field. Then $A$ is semi-involutory if and only if there exist non-singular diagonal matrices $D_1,D_2$ such that $D_1^n=k_1I$ and $D_2^n=k_2I$ for non-zero scalars  $k_1,k_2$ in the finite field, and $A^{-1}=D_1AD_2$. 
\end{theorem}

The analogous result for circulant semi-orthogonal matrices over a finite field is as follows.
\begin{theorem}\label{circ semi-ortho}
$A$ be an $n \times n$ circulant matrix over a finite field. Then $A$ is semi-orthogonal if and only if there exist non-singular diagonal matrices $D_1$ and $D_2$ such that $D_1^n=k_1I$ and $D_2^n=k_2I$ for non-zero scalars  $k_1,k_2 \in \mathbb{F}$ and $A^{-T}=D_1AD_2$.
\end{theorem}

\section{\bf Circulant matrices with MDS and semi-orthogonal properties \label{sec:circ semi-ortho} }

In \cite{GR2,GR}, Gupta {\it{et al.}} proved that circulant orthogonal matrices of order $2^d \times 2^d$ can not be MDS over a finite field of characteristic $2$. After that, in $2023$, Chatterjee {\it{et al.}} \cite{TA} studied semi-involutory and semi-orthogonal properties of circulant matrices.
They showed that in a circulant semi-orthogonal matrix of order $n \times n$, the $n$-th power of the associated diagonal matrices are scaler matrices.

Leveraging this property, we establish that for circulant semi-orthogonal matrices of order $2^d \times 2^d$, the trace of the associated diagonal matrices are zero over a finite field of characteristic $2$.
 
 \begin{proposition}\label{tr 0 2^d s.o}
 Let $A$ be a circulant, semi-orthogonal matrix of order $2^d \times 2^d$ over the finite field $\mathbb{F}_{2^m}$ with associated diagonal matrices $D_1$ and $D_2$. Then trace of  $D_1$ and $D_2$ are zero.
 \end{proposition}
 
 \begin{proof}
 Let $A$ be a circulant semi-orthogonal matrix with associated diagonal matrices $D_1$ and $D_2$. Then $A^{-T} = D_1AD_2,$ where $D_1$ and $D_2$ are non-singular diagonal matrices. Let $D_1=$ diagonal$(d_0,d_1,d_2,\hdots,d_{2^d-1})$ and $D_2=$ diagonal$(e_0,e_1,e_2,\hdots,e_{2^d-1})$. These two diagonal matrices also satisfy $D_1^{2^d}=k_1I$ and $D_2^{2^d}=k_2I$ for some non-zero scalers $k_1, k_2$ of the finite field by Theorem \ref{circ semi-ortho}. This implies trace$(D_1^{2^d})=2^dk_1=0$ and  trace$(D_2^{2^d})=2^dk_2=0$. This leads to the expressions: $$d_0^{2^d}+d_1^{2^d}+d_2^{2^d}+\cdots+d_{2^d-1}^{2^d}=(d_0+d_1+d_2+\cdots+d_{2^d-1})^{2^d}=0$$ and  $$e_0^{2^d}+e_1^{2^d}+e_2^{2^d}+\cdots+e_{2^d-1}^{2^d}=(e_0+e_1+e_2+\cdots+e_{2^d-1})^{2^d}=0.$$
 
Thus trace$(D_1)$ and trace$(D_2)$ are zero.
\end{proof}

For circulant, semi-orthogonal matrices of even orders other than powers of $2$, the following two theorems establish a significant relationship between the MDS property and the trace of the associated diagonal matrices.

\begin{theorem}
Let $A$ be a circulant, semi-orthogonal matrix of order $2^in \times 2^in$ over the finite field $\mathbb{F}_{2^m}$ with associated diagonal matrices $D_1$ and $D_2$, where $i >1$ and $n \geq 3,$ an odd integer.
Then $A$ is MDS implies both the matrices $D_1$ and $D_2$ have trace zero.
\end{theorem}

\begin{proof}
Let $2^in=k$ and $A=$ circulant$(a_0,a_1,a_2,\hdots,a_{k-1})$. Since $A$ is semi-orthogonal, we have $A^{-T} = D_1AD_2,$ where $D_1$ and $D_2$ are non-singular diagonal matrices given by $D_1=$ diagonal$(d_0,d_1,d_2,\hdots,d_{k-1})$ and $D_2=$ diagonal$(e_0,e_1,e_2,\hdots,e_{k-1})$. Let $A$ be an MDS matrix, then all the submatrices of $A$ have determinant non-zero.
Using the identity $AA^{-1}=I$, we have $AD_2A^{T}D_1=I$. Let $M=AD_2A^{T}D_1$. Since all non-diagonal entries of $M$ are zero, we can derive the following set of equations from the entries $M[0,1], M[1,2], M[2,3], \hdots, M[k-2,k-1], M[k-1,0]$:
\begin{align*}
(\sum_{i=0}^{k-1}a_ia_{i+1}&e_{i+1})d_1=0\\
(\sum_{i=0}^{k-1}a_ia_{i+1}&e_{i+2})d_2=0\\
&\vdots\\
(\sum_{i=0}^{k-1}a_ia_{i+1}&e_{i+(k-1)})d_{k-1}=0\\
(\sum_{i=0}^{k-1}a_ia_{i+1}&e_{i+k})d_0=0.
   \end{align*}
Here all the suffixes are calculated modulo $k$. Since $d_i$'s are non-zero, these equations reduce to the following:
\begin{eqnarray*}
(\sum_{i=0}^{k-1}a_ia_{i+1}e_{i+1})=0, (\sum_{i=0}^{k-1}a_ia_{i+1}e_{i+2})=0, \hdots,(\sum_{i=0}^{k-1}a_ia_{i+1}e_{i+k})=0
\end{eqnarray*}
Adding these equations we get
\begin{eqnarray}
(\sum\limits_{i=0}^{k-1}a_ia_{i+1})(e_0+e_1+\cdots+ e_{k-1})=0.
\end{eqnarray}

Next, consider the following set of entries of the matrix $M$: 
$M[0,3],$~ $M[1,4],$~ $M[2,5],$ $\hdots,$ $M[k-3,0],$ $M[k-2,1],$ $M[k-1,2]$. From these entries, we get the following set of equations:
\begin{align*}
(\sum_{i=0}^{k-1}a_ia_{i+3}&e_{i+3})d_3=0\\
(\sum_{i=0}^{k-1}a_ia_{i+3}&e_{i+4})d_4=0\\
(\sum_{i=0}^{k-1}a_ia_{i+3}&e_{i+5})d_5=0\\
&\vdots\\
(\sum_{i=0}^{k-1}a_ia_{i+3}&e_{i+k+1})d_1=0\\
(\sum_{i=0}^{k-1}a_ia_{i+3}&e_{i+k+2})d_2=0.
   \end{align*}
Here all the suffixes are calculated modulo $k$. Similarly as before, using that $d_i$'s are non-zero and adding these equations, we get
\begin{eqnarray}
(\sum_{i=0}^{k-1}a_ia_{i+3})(e_0+e_1+\cdots+ e_{k-1})=0.
\end{eqnarray}
Continuing this process to cover all the odd positions of the first row till the position $M[0,\frac{k}{2}]$.
 
Consider the entries at positions $ M[0,\frac{k}{2}-1], M[1,\frac{k}{2}], M[2,\frac{k}{2}+1], \hdots, M[\frac{k}{2}+1,0], M[\frac{k}{2}+2,1], \cdots, M[k-1,\frac{k}{2}-2]$. From these entries, we get the equation 
\begin{eqnarray}
(\sum\limits_{i=0}^{k-1}a_ia_{i+{\frac{k}{2}-1}})(\sum\limits_{i=0}^{k-1}e_i)=0,
\end{eqnarray} 
where the suffixes are calculated modulo $k$.

Adding the following $\frac{k}{4}$ equations 
\begin{eqnarray*}
(\sum\limits_{i=0}^{k-1}a_ia_{i+1})(\sum\limits_{i=0}^{k-1}e_i)=0, (\sum\limits_{i=0}^{k-1}a_ia_{i+3})(\sum\limits_{i=0}^{k-1}e_i)=0, \hdots, (\sum\limits_{i=0}^{k-1}a_ia_{i+{\frac{k}{2}-1}})(\sum\limits_{i=0}^{k-1}e_i)=0,
\end{eqnarray*}
we get
\begin{eqnarray}\label{final eq1}
(a_0+a_2+\cdots+a_{k-2})(a_1+a_3+\cdots+a_{k-1})(e_0+e_1+\cdots+ e_{k-1})=0.
\end{eqnarray}

Given that $A$ is a circulant matrix of order $k \times k$, it has two circulant submatrices of order $\frac{k}{2}$ with the first row $(a_0,a_2,\hdots,a_{k-2})$ and $(a_1,a_3,\hdots,a_{k-1})$ respectively. According to Equation (\ref{det circulant})
, both $(a_0+a_2+\cdots+a_{k-2})$ and $(a_1+a_3+\cdots+a_{k-1})$ must be non-zero since $A$ is an MDS matrix. Therefore from Equation \ref{final eq1}, we have $(e_0+e_1+\cdots+ e_{k-1})=0$ and this implies trace$(D_2)=0$.

Similarly using the identity $A^{-1}A=I$ and following the same process, we will get trace$(D_1)=0$. 
\end{proof}

In the next result, we explore the case where the order of the matrix is an even number of the form $2n, n$ is an odd number. In this case, we need one additional condition on the entries of at least one of the associated diagonal matrix. Any diagonal matrix of even order  meeting this criterion is termed as non-periodic diagonal matrix. Specifically, we define a diagonal matrix $D=$ diagonal $(d_0,d_1,d_2,\hdots,d_{2n-1})$ as a non-periodic diagonal matrix, if the entries satisfy  $d_i \neq d_{i+n}, i=0,1,2, \hdots, n-1$.

\begin{theorem}
Let $A$ be a circulant, semi-orthogonal matrix of order $2n \times 2n, n \geq 3$ be an odd number, over $\mathbb{F}_{2^m}$ with associated diagonal matrices $D_1$ and $D_2$. If $A$ is an MDS matrix and at least one of the associated diagonal matrix is non-periodic, then trace of that non-periodic diagonal matrix is zero.
\end{theorem}

\begin{proof}
Let $A=$circulant$(a_0,a_1,a_2,\hdots,a_{2n-1})$. Since $A$ is semi-orthogonal, it satisfy $A^{-T}=D_1AD_2$, where $D_1$ and $D_2$ are non-singular diagonal matrices given by $D_1=$diagonal$(d_0,d_1,d_2,\hdots,d_{2n-1})$ and $D_2=$diagonal$(e_0,e_1,e_2,\hdots,e_{2n-1})$. Without loss of generality, we assume that $D_2$ is non-periodic diagonal matrix. Then $e_i \neq e_{i+n}, i=0,1,2, \hdots, n-1.$ Let $A$ be an MDS matrix.

Since $AA^{-1}=I$, we have $AD_2A^{T}D_1=I$. Let $AD_2A^{T}D_1=M$. All non-diagonal entries of $M$ are zero. Form the entries $M[0,1], M[1,2], M[2,3], \hdots, M[2n-2,2n-1],  M[2n-1,0]$ we get the following equations:
\begin{align*}
(\sum_{i=0}^{2n-1}a_ia_{i+1}&e_{i+1})d_1=0\\
(\sum_{i=0}^{2n-1}a_ia_{i+1}&e_{i+2})d_2=0\\
&\vdots\\
(\sum_{i=0}^{2n-1}a_ia_{i+1}&e_{i+(2n-1)})d_{2n-1}=0\\
(\sum_{i=0}^{2n-1}a_ia_{i+1}&e_{i+2n})d_0=0.
 \end{align*}
Here all the suffixes are calculated modulo $2n$.
Since $d_i$'s are non-zero, we can add these equations and get
\begin{eqnarray*}
(\sum_{i=0}^{2n-1}a_ia_{i+1})(e_0+e_1+\cdots+ e_{2n-1})=0.
\end{eqnarray*}

Continuing the similar process for the entries at positions $M[0,3],$ $M[1,4],$ $M[2,5],$ $\hdots,$ $M[2n-3,0]$ we get:
\begin{align*}
(\sum_{i=0}^{2n-1}a_ia_{i+3}&e_{i+3})d_3=0\\
(\sum_{i=0}^{2n-1}a_ia_{i+3}&e_{i+4})d_4=0\\
&\vdots\\
(\sum_{i=0}^{2n-1}a_ia_{i+3}&e_{i+(2n)})d_{0}=0\\
(\sum_{i=0}^{2n-1}a_ia_{i+3}&e_{i+2n+1})d_1=0\\
(\sum_{i=0}^{2n-1}a_ia_{i+3}&e_{i+2n+2})d_2=0.
\end{align*}
Here all the suffixes are calculated modulo $2n$. Adding these equations we get $$(\sum\limits_{i=0}^{2n-1}a_ia_{i+3})(\sum\limits_{i=0}^{2n-1}e_i )= 0.$$

Continue this process to cover all the odd positions of the first row, upto the position $M[0,n]$.
From the entries $M[0,n], M[1,n+1], M[2,n+2], \hdots,M[n-1,2n-1]$ we get:
\begin{equation}\label{last set}
    \begin{aligned}
(\sum_{i=0}^{n-1}a_ia_{i+n}&(e_i+e_{i+n}))d_n=0\\
(\sum_{i=0}^{n-1}a_ia_{i+n}&(e_{i+1}+e_{i+(n+1)}))d_{n+1}=0\\
&\vdots\\
(\sum_{i=0}^{n-1}a_ia_{i+n}&(e_{i+(n-1)}+e_{(i+n)+(n-1)}))d_{2n-1}=0.
\end{aligned}
  \end{equation}
Here all the suffixes are calculated modulo $2n$. Using the given conditions on $e_i$'s and $d_i$'s non-zero, we get $(\sum\limits_{i=0}^{2n-1}a_ia_{i+n})(\sum\limits_{i=0}^{2n-1}e_i )=0$.

 Note that, in Equation (\ref{last set}), we have $n$ number of equations, where the other sets of involve $2n$ equations each. 
 Finally, adding the following $\lceil \frac{n}{2} \rceil$ equations: $$(\sum\limits_{i=0}^{2n-1}a_ia_{i+1})(\sum\limits_{i=0}^{2n-1}e_i )=0, (\sum\limits_{i=0}^{2n-1}a_ia_{i+3})(\sum\limits_{i=0}^{2n-1}e_i )=0, \hdots, 
 (\sum\limits_{i=0}^{n-1}a_ia_{i+n})(\sum\limits_{i=0}^{2n-1}e_i )=0,$$ we obtain $$(a_0+a_2+\cdots+a_{2n-2})(a_1+a_3+\cdots+a_{2n-1})(e_0+e_1+\cdots+ e_{2n-1})=0.$$ Since $A$ is MDS, using the same argument as previous theorem, we get $(\sum\limits_{i=0}^{2n-1}e_i )=0$. This implies trace$(D_2)=0.$
 
Similarly using the identity $A^{-1}A=I,~ D_1$ is non-cyclic diagonal matrix, and following the same process, we will get trace$(D_1)=0$.
\end{proof}

For circulant, semi-orthogonal matrices of odd order, the following examples demonstrate the possibility of achieving the MDS property and the trace of the associated diagonal matrices are non-zero.
 \begin{example}
Consider the $3 \times 3$ matrix $A=$ circulant $(\alpha, \alpha+1, \alpha^2+\alpha),$ where $\alpha$ is a primitive element of the finite field $\mathbb{F}_{2^8}$ with the generating polynomial $x^8+x^4+x^3+x^2+1$. Note that, $A$ is semi-orthogonal since $A^{-T}=D_1AD_2,$ where $D_1=$ diagonal$(\alpha^7+\alpha^6+\alpha^5+\alpha, \alpha^7+\alpha^6+\alpha^5+\alpha, \alpha^7+\alpha^6+\alpha^5+\alpha)$ and $D_2=$ diagonal$(\alpha^6+\alpha^4+\alpha^3+\alpha, \alpha^6+\alpha^4+\alpha^3+\alpha, \alpha^6+\alpha^4+\alpha^3+\alpha)$. $A$ is also an MDS matrix.
 \end{example}
 \begin{example}
Consider the $5 \times 5$ matrix $A=$ circulant $(1,1+\alpha+\alpha^3,1+\alpha+\alpha^3,\alpha+\alpha^3,1+\alpha^3+\alpha^4+\alpha^7),$ where $\alpha$ is a primitive element of the finite field $\mathbb{F}_{2^8}$ with the generating polynomial $x^8+x^4+x^3+x^2+1$.  
Note that, $A$ is semi-orthogonal since $A^{-T}=D_1AD_2,$ where $D_1=$ diagonal$(\alpha^2+\alpha, \alpha^7+\alpha^2+1, \alpha^7+\alpha^6+\alpha^5+\alpha^4+\alpha^2,\alpha^5+\alpha^4+\alpha^3+\alpha^2,\alpha^6+\alpha^3+\alpha+1)$ and $D_2=$ diagonal$(\alpha^7+\alpha^6+\alpha^3+\alpha^2+\alpha+1, \alpha^7+\alpha^5+\alpha^3, \alpha^7+\alpha^5+\alpha^4+\alpha^2+1, \alpha^6+\alpha^5+\alpha^2, \alpha^7+\alpha^5+\alpha^4+\alpha^2+\alpha)$. $A$ is also an MDS matrix.

\end{example}
 \begin{remark}
We have classified circulant semi-orthogonal matrices over the finite field $\mathbb{F}_{2^m}$ into four distinct categories. Specifically, for odd orders, we provide examples of circulant semi-orthogonal matrices of orders $3 \times 3$ and $5 \times 5$ with the MDS property. For matrices of  order $2^d \times 2^d$, the trace of the associated diagonal matrices is zero. Additionally, for matrices of even order, where the order $k \equiv 0 \pmod 4$, the MDS property ensures that the trace of the associated diagonal matrices remains zero. Furthermore, when the order is even and congruent to $2 \pmod 4$, the MDS property together with non-periodic diagonal matrices results in a trace value of zero for the associated diagonal matrices.
\end{remark} 
 
 In the subsequent section, we explore circulant matrices with the semi-involutory property. Our objective is to determine whether similar outcomes persist under semi-involutory property or not.
 
 \section{\bf Circulant matrices with MDS and semi-involutory properties \label{sec:circ semi-inv} }

In \cite{GR2}, Gupta {\it{et al.}} proved that circulant involutory matrices of order $n \geq 3$ can not be MDS. In the subsequent results, we extend this characteristic to circulant semi-involutory matrices. In this direction, our first result demonstrates that, the trace of these associate diagonal matrices is also zero under certain conditions.

\begin{proposition}\label{tr 0 2^d si}
Let $A$ be a $2^d \times 2^d$ circulant, semi-involutory matrix over the finite field $\mathbb{F}_{2^m}$ with associated diagonal matrices $D_1$ and $D_2$. Then trace of $D_1$ and $D_2$ are zero.
\end{proposition}

\begin{proof}
The proof follows similarly as Theorem \ref{tr 0 2^d s.o}.\end{proof}

For circulant semi-involutory matrices with orders other than $2^d \times 2^d$, our result establishes that the trace value of the associated diagonal matrices is zero when the matrix exhibits the MDS property.

\begin{theorem}
Let $A$ be an $n \times n, n \geq 3, n \neq 2^i$ circulant, semi-involutory matrix over the finite field $\mathbb{F}_{2^m}$ with associated diagonal matrices $D_1$ and $D_2$. Then $A$ is MDS implies both the matrices $D_1$ and $D_2$ have trace zero.

\end{theorem}
\begin{proof}
Let $A=$ circulant$(a_0,a_1,a_2,\hdots,a_{n-1})$. Since $A$ is semi-involutory, we have $A^{-1} = D_1AD_2,$ where $D_1$ and $D_2$ are non-singular diagonal matrices given by $D_1=$ diagonal$(d_0,d_1,d_2,\hdots,d_{n-1})$ and $D_2=$ diagonal$(e_0,e_1,e_2,\hdots,e_{n-1})$. 

Let $A$ be an MDS matrix. Since $AA^{-1}=I$, we have $AD_1AD_2=I$. Let $M=AD_1AD_2$.  This implies that all non-diagonal entries of $M$ are $0$. 

\textbf{Case I:} Consider the case $n$ is even, $n=2k$.

From the entries $M[0,2], M[1,3], M[2,4], M[3,5],\hdots, M[2k-3,2k-1], M[2k-2,0], M[2k-1,1]$ we get the following equations:
\begin{align*}
(a_1^2d_1+&a_{k+1}^2d_{k+1}+a_0a_2(d_0+d_2)+a_3a_{2k-1}(d_3+d_{2k-1})+\cdots+a_ka_{k+2}(d_k+d_{k+2}))e_2=0\\
(a_1^2d_2+&a_{k+1}^2d_{k+2}+a_0a_2(d_1+d_3)+a_3a_{2k-1}(d_4+d_{2k})+\cdots+a_ka_{k+2}(d_{k+1}+d_{k+3}))e_3=0\\
&\vdots\\
(a_1^2d_{2k-2}&+a_{k+1}^2d_{k-2}+a_0a_2(d_{2k-1}+d_{2k-3})++a_3a_{2k-1}(d_{3+(2k-3)}+d_{2k-4})+\cdots+a_ka_{k+2}\\&(d_{k-1}+d_{k-3}))e_{2k-1}=0\\
(a_1^2d_{2k-1}&+a_{k+1}^2d_{k-1}+a_0a_2(d_{2k}+d_{2k-2})+a_3a_{2k-1}(d_1+d_{2k-3})+\cdots+a_ka_{k+2}\\&(d_{k-2}+d_{k}))e_0=0\\
(a_1^2d_{2k}&+a_{k+1}^2d_{k}+a_0a_2(d_{1}+d_{2k-1})+a_3a_{2k-1}(d_2+d_{2k-2})+\cdots+a_ka_{k+2}\\&(d_{k-1}+d_{k+1}))e_1=0.
\end{align*}
All the suffixes are calculated modulo $2k$. Since $e_i$'s are non-zero, adding all these equations, we get  
\begin{eqnarray}\label{fin eq 3}
  (a_1^2+a_{k+1}^2)(d_1+d_2+\cdots+ d_{2k-1})=0
\end{eqnarray}
Since $A$ is an MDS matrix, all its submatrices have determinant non-zero. Consider the $2 \times 2$ submatrix of $A$ with the positions $A[0,1], A[0,k+1], A[k,1], A[k,k+1]$. The determinant of this submatrix is $(a_1^2+a_{k+1}^2)$ and thus it is non-zero. Consequently Equation (\ref{fin eq 3}) implies $(d_1+d_2+\cdots+ d_{2k-1})=0$. Therefore trace of $D_1$ is zero.

Considering the identity $A^{-1}A=I$ and proceed similarly, we will get trace of $D_2$ is zero.


\textbf{Case II:} Consider the case $n$ is odd, $n=2k+1$.

The entries at the positions $M[0,2], M[1,3], M[2,4], M[3,5], \hdots, M[2k-2,2k], M[2k-1,0], M[2k,2]$ give the following equations:
\begin{equation*}
    \begin{aligned}
(a_1^2d_1+&a_0a_2(d_0+d_2)+a_3a_{2k}(d_3+d_{2k})+\cdots+a_{k+1}a_{k+2}(d_{k+1}+d_{k+2}))e_2=0\\
(a_1^2d_2+&a_0a_2(d_1+d_3)+a_3a_{2k}(d_4+d_{2k+1})+\cdots+a_{k+1}a_{k+2}(d_{k+2}+d_{k+3}))e_3=0\\
&\vdots\\
(a_1^2d_{2k-1}&+a_0a_2(d_{0+2k-2}+d_{1+2k-2})+a_3a_{2k}(d_{3+2k-2}+d_{2k+2k-2})+\cdots+a_{k+1}a_{k+2}\\&(d_{k+1+2k-2}+d_{k+2++2k-2}))e_{2k-1}=0\\
(a_1^2d_{2k}&+a_0a_2(d_{0+2k-1}+d_{2+2k-1})+a_3a_{2k}(d_{3+2k-1}+d_{2k+2k-1})+\cdots+a_{k+1}a_{k+2}\\&(d_{k+1+2k-1}+d_{k+2+2k-1}))e_0=0\\
(a_1^2d_0&+a_0a_2(d_1+d_{2k})+a_3a_{2k}(d_2+d_{2k-1})+\cdots+a_{k+1}a_{k+2}(d_{k}+d_{k+1}))e_1=0.\\
\end{aligned}
  \end{equation*}
All the suffixes are calculated modulo $2k+1$. Since $e_i$'s are non-zero, adding all these equations, we get  
  \begin{eqnarray}
  (a_1^2)(d_1+d_2+\cdots+d_{2k})=0
  \end{eqnarray}
  Since $A$ is an MDS matrix, all entries of $A$ are non-zero. This implies $(d_1+d_2+\cdots+d_{2k})=0$. Therefore trace of $D_1$ is zero.

Considering the identity $A^{-1}A=I$ and proceed similarly, we will get trace of $D_2$ is zero. 

\end{proof}

\begin{remark}
For circulant semi-involutory matrices over the finite field $\mathbb{F}_{2^m}$, we have proven that matrices of order $2^d \times 2^d$ exhibit a zero trace for their associated diagonal matrices. Furthermore, for orders not represented as powers of $2$, the trace remains zero if the matrix possesses the MDS property.
\end{remark}

\section{Conclusion} 
In conclusion, this article has explored circulant matrices with both semi-orthogonal and MDS properties, as well as circulant matrices characterized by semi-involutory and MDS attributes. We have presented examples of circulant semi-orthogonal MDS matrices for certain odd orders. Moreover, our results open a new research direction to investigate circulant MDS matrices with both semi-involutory and semi-orthogonal properties. This is particularly significant given the comprehensive examination of circulant involutory and orthogonal matrices in existing literature.

\bibliographystyle{plain}

\end{document}